\documentclass[11pt]{article}

\usepackage[
	persons={Nima, Sinho, June},
	authors=footnotes,
    bibliosources={refs.bib}
]{pomegranate}

\usepackage{algorithm}
\usepackage{algpseudocode}

\title{Fast parallel sampling under isoperimetry}
\date{\today}

\Tag<anon>{
	\author{}
}
\Tag{
	\author[1]{Nima Anari}
	\author[2]{Sinho Chewi}
	\author[1]{Thuy-Duong Vuong}
	\affil[1]{Stanford University, \url{{anari,tdvuong}@stanford.edu}}
	\affil[2]{Institute for Advanced Study, \url{schewi@ias.edu}}
}

\allowdisplaybreaks{}
\DeclareTheorem<definition>{assumption}

\newcommand{\deq}{\coloneqq}
\newcommand{\E}{\operatorname{\mathbb E}}
\newcommand{\Ex}{\operatorname{\mathbb E}}
\renewcommand{\given}{\;|\;}
\renewcommand{\medskip}{}
\newcommand{\mb}[1]{\mathbf #1}
\newcommand{\mf}[1]{\mathfrak #1}
\newcommand{\mc}[1]{\mathcal #1}
\DeclareMathOperator{\law}{law}
\DeclareMathOperator{\polylog}{poly\,log}
\DeclareDelimiter{\Normal}[\mathcal{N}]{\lparen}{\rparen}
\DeclareDelimiter{\L}[\mathcal{L}]{\lparen}{\rparen}
\DeclareDelimiter{\sign}[sign]{\lparen}{\rparen}
\DeclareDelimiter{\mean}[mean]{\lparen}{\rparen}
\DeclareDelimiter{\cov}[cov]{\lparen}{\rparen}
\DeclareDelimiter{\FI}[FI]{\lparen}{\rparen}
\DeclareDelimiter{\Ent}[Ent]{\lbrack}{\rbrack}
\DeclareDelimiter{\Otilde}[\mathnormal{\widetilde{O}}]{\lparen}{\rparen}
\DeclareDelimiter{\DKL}[\mathcal{D}_{\operatorname{KL}}]{\lparen}{\rparen}
\DeclareDelimiter{\dTV}[\mathnormal{d}_{\operatorname{TV}}]{\lparen}{\rparen}

\begin{document}
	\maketitle
	\begin{abstract}
	We show how to sample in parallel from a distribution $\pi$ over $\mathbb{R}^d$ that satisfies a log-Sobolev inequality and has a smooth log-density, by parallelizing the Langevin (resp.\ underdamped Langevin) algorithms. We show that our algorithm outputs samples from a distribution $\hat{\pi}$ that is close to $\pi$ in Kullback--Leibler (KL) divergence (resp.\ total variation (TV) distance), while using only $\log(d)^{O(1)}$ parallel rounds and $\widetilde{O}(d)$ (resp.\ $\widetilde O(\sqrt d)$) gradient evaluations in total. This constitutes the first parallel sampling algorithms with TV distance guarantees.
	
	For our main application, we show how to combine the TV distance guarantees of our algorithms with prior works and obtain RNC sampling-to-counting reductions for families of discrete distribution on the hypercube $\{\pm 1\}^n$  that are closed under exponential tilts and have bounded covariance. Consequently, we obtain an RNC sampler for directed Eulerian tours and asymmetric determinantal point processes, resolving open questions raised in prior works. 
	\end{abstract}
 
	\section{Introduction}

In this paper, we study the problem of designing fast parallel algorithms for sampling from continuous distributions $\pi(x)\propto \exp(-V(x))$ over $x\in \R^d$. Designing efficient sampling algorithms is a ubiquitous problem, but the focus of most prior works has been to minimize sequential efficiency criteria, such as the total number of arithmetic operations or total queries to $V$ and its derivatives (see~\cite{Chewi23Book} for an exposition). In contrast, in this work we focus on parallel efficiency; roughly speaking, this means that we would like to have algorithms that sequentially take polynomial time, but can be run on a pool of polynomially many processors (e.g., as in the \Class{PRAM} model of computation) in much less time, ideally polylogarithmic.

Our main result is to propose simple parallelizations of Langevin Monte Carlo (LMC) and underdamped Langevin Monte Carlo (ULMC), two of the most widely studied sequential sampling algorithms, and to prove that they run in $\log(d)^{O(1)}$ parallel iterations, under standard tractability criteria on $\pi$: that it satisfies a log-Sobolev inequality (LSI), and that its potential $V$ is smooth, i.e., has Lipschitz gradients.

\begin{theorem}[Informal main theorem]\label{thm:informal-main}
    Suppose that $\pi=\exp(-V)$ is a density on $\R^d$ that satisfies a log-Sobolev inequality and has a smooth potential $V$. Assume that we are given (approximate) oracle access to $\nabla V$. Then, we can produce samples from a distribution $\hat{\pi}$ with the following guarantees.
    \begin{itemize}
         \item For LMC, $\hat\pi$ is close to $\pi$ in Kullback--Leibler divergence, and the algorithm uses $\log^2(d)$ parallel iterations and $\widetilde O(d)$ processors and gradient evaluations.
         \item For ULMC, $\hat\pi$ is close to $\pi$ in total variation divergence, and the algorithm uses $\log^2(d)$ parallel iterations and $\widetilde O(\sqrt d)$ processors and gradient evaluations.
    \end{itemize}
\end{theorem}

For formal statements, see~\Cref{thm:sample-main} and~\Cref{thm:ulmc_discretization}.
Throughout this paper, when we refer to the number of iterations, we refer to the model of \emph{adaptive complexity}: here, in each round, the algorithm makes a batch of queries to a first-order oracle for $\pi$ (i.e., given a set of finite points $\mathcal X \subseteq \R^d$, the oracle outputs $(V(x), \nabla V(x))$ for each $x\in\mathcal X$), and the adaptive complexity measures the number of rounds.
The gradient complexity measures the total number of points at which the first-order oracle is queried.

As an immediate corollary, we obtain parallel samplers for the class of well-conditioned log-concave distributions, i.e., those which satisfy 
\[ \beta I\succeq \nabla^2 V \succeq \alpha I\,,\]
for some constants $\alpha, \beta>0$, where $\beta$ is the smoothness parameter, and $\alpha$ is the parameter of strong log-concavity. This is because the LSI, a form of isoperimetric inequality, holds for all strongly log-concave distributions, due to the Bakry--\'{E}mery criterion~\citep{BE06}. However, the LSI is a weaker condition than strong log-concavity, and it applies to even many non-log-concave distributions such as Gaussian convolutions of distributions with bounded support~\citep{bardet2015functional,chen2021dimensionfree}. In addition, unlike log-concavity, LSI is preserved under bounded perturbations and Lipschitz transformations of the log-density function.

The state-of-the-art prior to our work was a fast parallel algorithm due to \cite{SL19}, which produced Wasserstein-approximate samples from well-conditioned log-concave distributions. We improve on the state-of-the-art in three ways:
\begin{itemize}
	\item We replace the strong log-concavity assumption with the weaker assumption that $\pi$ satisfies a log-Sobolev inequality.
	\item We bound the error in KL divergence and TV distance, as opposed to the weaker notion of Wasserstein error. This difference is crucial for our main application, as explained in~\Cref{sec:application}.
	\item Our results hold given only approximate access to $\nabla V$, as opposed to exact access. This is again crucial in some of our applications as explained in~\Cref{sec:application}.
\end{itemize}

\subsection{Algorithm}

For the sake of exposition, here we describe the parallel LMC algorithm and defer the discussion of parallel ULMC to~\Cref{sec:ulmc_alg}.

Our algorithm is based on a parallelized discretization of the Langevin diffusion. The continuous-time Langevin diffusion is the solution to the stochastic differential equation
\begin{equation} \label{eq:langevin diffusion}
    dX_t = -\nabla V(X_t)\,dt + \sqrt{2} \,d B_t
\end{equation}
where $(B_t)_{t\geq 0}$ is a standard Brownian motion in $\R^d.$
Langevin Monte Carlo (LMC) is a discretization of the continuous Langevin diffusion, defined by the following iteration:
\begin{equation}
  X_{(n+1) h} - X_{nh} = - h\,\nabla V(X_{nh}) + \sqrt{2}\, (B_{(n+1) h} - B_{nh})\,, 
\end{equation}
where $h > 0$ is a parameter defining the step size.

If $\pi$ satisfies a log-Sobolev inequality (LSI), then the law of the continuous-time Langevin diffusion converges to the target distribution $\pi$ at time $t\approx \polylog(d)$. The discretization error, measured for example in the total variation distance, between the continuous Langevin diffusion and the discrete process, scales like $\approx dh$, so the step size $h$ is set to $1/d$, causing LMC to take $\Otilde{d}$ iterations to converge. Our algorithm, explained in~\Cref{alg:main}, uses parallelization to speed up LMC, so that the step size is $\Omega(1)$ and the parallel depth is of the same order as the convergence time of the continuous Langevin diffusion, that is, of order $\polylog(d)$.  

The input to the algorithm is a (potentially random) starting point $X_0$, together with an ``approximate score oracle'' $s$, which is a function $\R^d\to \R^d$ that we can query, and which is assumed to be uniformly close to the gradient $\nabla V$.

The main idea behind the algorithm is to turn the task of finding solutions to our (stochastic) differential equation into the task of finding fixed points of what is known as the Picard iteration. At a high level, Picard iteration takes a trajectory $(X_t)_{t\geq 0}$ and maps it to another trajectory $(X'_t)_{t\geq 0}$ given by
\[ X'_t = X_0 -\int_{0}^t \nabla V(X_u)\,du + \sqrt 2 \, B_t\,. \]
Now if $X=X'$, then $X$ is a solution to the Langevin diffusion. Thus, one might hope that starting from some trajectory $X_0$, and applying Picard iterations multiple times, the whole trajectory converges to the fixed point. The main benefit of Picard iteration is that $\nabla V$ or $s$ can be queried at all points in parallel.

Note that the Picard iteration can be analogously defined for discrete-time dynamics such as LMC\@. Our main result shows that Picard iteration applied to the discretized Langevin diffusion (LMC) converges fast (in $\polylog(d)$ Picard iterations) for trajectories defined over intervals of length at most $h$, where now $h$ can be take nto be macroscopically large ($h = \Omega(1)$). We repeat this process until time $\polylog(d)$, which requires $N=\polylog(d)/h$ sequential iterations.

\begin{algorithm}[t]
\textbf{Input}: $X_0\sim \mu_0$, approximate score function $s: \R^d \to \R^d$ ($s\approx \nabla V$) \;

\For{$n=0, \dots, N-1$}{
    \For{$m=0, \dots, M$ in parallel}{
        $X_{nh+mh/M}^{(0)} \gets X_{nh}$\;

        Sample Brownian motion $B_{nh+mh/M} \leftarrow B_{nh} + \mathcal{N}(0,(mh/M)\,I)$
    }
    \For{$k=0, \dots, K-1$}{
        \For{$m=0, \dots, M$ in parallel}{
        $X_{nh+mh/M}^{(k+1)}\gets X_{nh}-\frac{h}{M} \sum_{m'=0}^{m-1} s(X_{nh+m'h/M}^{(k)}) + \sqrt{2}\, (B_{nh+mh/M} - B_{nh}) $
        }    
    }
    $X_{(n+1)h} \gets X_{nh+h}^{(K)}$
    \caption{Parallelized Langevin dynamics \label{alg:main}}	
}
\end{algorithm}

\subsection{Analysis techniques}\label{sec:techniques}

Many algorithms for solving stochastic differential equations, such as the Langevin dynamics $(X_t^*)_{t\ge 0}$, turn the problem into numerical integration. The main idea is to approximate the difference between $X_{(n+1) h}^*-X_{nh}^*$ using the trapezoidal rule, i.e.,
\begin{align*}
    X^*_{(n+1) h} - X^*_{nh}
    &= -\int_{nh}^{(n+1) h} \nabla V(X^*_s) \,ds + \sqrt 2 \, (B_{(n+1)h}-B_{nh}) \\
    &\approx -\sum w_i\, \nabla V (X^*_{s_i}) + \sqrt 2 \, (B_{(n+1)h}-B_{nh})\,.
\end{align*}
Since we cannot access the idealized process $X^*$, we instead start with a rough estimate $X^{(0)}$ and iteratively refine our estimation to obtain $X^{(1)}, \dots, X^{(K)}$ that are closer and closer to the ideal $X^*$. The refined estimations are obtained via another application of the trapezoidal rule, i.e., $X^{(k)} _{s_{i}}$ is computed using $\int_{s\leq s_i} \nabla V (X^{(k-1)}_s)\, ds$. This framework can be easily parallelized: $\nabla V(X^{(k)}_{s_i})$ for different $i$'s can be computed in parallel using one processor for each $s_i$. 

In \cite{SL19}, the points $s_i$ at which to evaluate $\nabla V (X_{s_i})$ are chosen randomly; hence, their framework is known as the \text{randomized midpoint} method. Unfortunately, there seem to be fundamental barriers to obtaining KL or TV accuracy guarantees for randomized midpoint algorithms. To illustrate, while accuracy in $2$-Wasserstein distance can be achieved using $\widetilde O(d^{1/3})$ gradient evaluations using a randomized midpoint algorithm \cite[Algorithm 1]{SL19}, accuracy in KL or TV distance using $o(d^{1/2})$ gradient evaluations is not known.

We deviate from the approach of \textcite{SL19} by keeping the $s_i$ fixed. This greatly simplifies the algorithm and its analysis and allows us to show that parallelized LMC converges to $\pi$ in KL divergence using the interpolation method~\citep{VW19}, at the cost of using $\widetilde O(d)$ gradient evaluations instead of $\widetilde O(\sqrt{d})$\footnote{While~\citet[Algorithm 1]{SL19} needs only $\widetilde O(d^{1/3})$ gradient evaluations, its parallel round complexity is also $\widetilde\Theta(d^{1/3})$, which doesn't align with our goal of getting $\polylog(d)$ parallel round complexity. On the other hand,~\citet[Algorithm 2]{SL19} uses $\polylog(d)$ parallel rounds but needs $\widetilde\Theta(\sqrt{d})$ gradient evaluations \cite[see][Theorem 4]{SL19}.} as in~\citet[Algorithm 2]{SL19}. 
In~\Cref{sec:ulmc}, we then show how to obtain a sampler, based on ULMC\@, which enjoys the same parallel complexity but uses only $\widetilde O(\sqrt d)$ gradient evaluations, matching the state-of-the-art in~\cite{SL19}.

For simplicity of exposition, assume that in~\Cref{alg:main}, the score function $s$ is exactly $\nabla V$.
We will show via induction that 
\begin{equation}\label{ineq:refine gradient bound}
    \E[\|\nabla V(X_{s_i}^{(K)}) - \nabla V(X^{(K-1)}_{s_i})\|^2]
  \lesssim \exp(-3.5K)\,,
\end{equation}
where $K$ is the depth of refinement. In other words, the approximation error decays exponentially fast with the parallel depth.

To obtain the KL divergence bound, note that
\[ X^{(K)}_{nh+(m+1)h/M} - X^{(K)}_{nh+m h/M} = -\frac{h}{M}\,\nabla V(X^{(K-1)}_{nh+mh/M}) + \sqrt{2}\, (B_{nh+(m+1)h/M} - B_{nh+m h/M})  \]
and $\nabla V(X^{(K-1)}_{nh+mh/M})$ only depends on $ X_{nh}$ and the Brownian motion $B_t$ for $t\leq mh/M.$ Let $X_t$, $mh \le t-nh \le (m+1)h$, be the interpolation of $X_{nh+m h/M}^{(K)}$ and $X_{nh+(m+1) h/M}^{(K)}$, i.e.,
\begin{equation*}
  X_t - X_{nh+mh/M}^{(K)} = - (t-nh-mh/M)\, \nabla V (X_{nh+mh/M}^{(K-1)}) +\sqrt{2} \,(B_t -B_{nh+mh/M})\,.
\end{equation*}
Then by a similar argument as in \cite{VW19}, if $\mu_t \deq \law(X_t^{(K)})$ we obtain
\begin{align*}
  \partial_t \DKL{\mu_t \river\pi}  &\leq -\frac{3 \alpha }{2 } \DKL{\mu_t \river\pi} + \E[\|\nabla V(X_t^{(K)}) - \nabla V (X_{nh+mh/M}^{(K-1)}) \|^2]  \\
  &\leq -\frac{3 \alpha }{2 } \DKL{\mu_t \river\pi}
  + 2 \E[\|\nabla V(X_t^{(K)}) - \nabla V (X_{nh+mh/M}^{(K)}) \|^2] \\
      &\qquad{}+ 2\E[\|\nabla V(X_{nh+mh/M}^{(K)}) - \nabla V (X_{nh+mh/M}^{(K-1)}) \|^2]\,.
\end{align*}
We can directly bound the third term using~\cref{ineq:refine gradient bound}. The second term can be bounded via a standard discretization analysis, noting that the time interval is only of size $h/M$.
It leads to the bound
\begin{align}\label{eq:intro_disc_bd}
    \E[\|\nabla V(X_t^{(K)}) - \nabla V (X_{nh+mh/M}^{(K)}) \|^2] \lesssim \frac{dh}{M}\,,
\end{align}
where $M$ is the number of discretization points, i.e., the number of parallel score queries in each round.
Thus, from~\cref{ineq:refine gradient bound} and~\cref{eq:intro_disc_bd}, by setting $K= \Otilde{1}$ and $M=\Otilde{d}$, we can set the step size $h= \Omega (1)$ so that the parallelized Langevin algorithm takes $\Otilde{1}$ steps to converge to the target distribution $\pi$. 

\begin{remark}
	One may wonder if our results apply to distributions satisfying a weaker functional inequality such as the Poincar\'e inequality, instead of the LSI. Unfortunately, this is not the case since our analysis relies on the fact that the continuous-time Langevin diffusion converges to the target distribution $\pi$ in time $\polylog (d)$, which holds under the LSI but not under the weaker Poincar\'e inequality \cite[see][for details]{chewi2021analysis}.	
\end{remark}


The above strategy based on the interpolation method no longer works for ULMC\@, so here we instead use an approach based on Girsanov's theorem. See~\Cref{sec:ulmc_analysis} for details.

\subsection{Applications}\label{sec:application}

The main application of our results is to obtain fast parallel algorithms for several discrete sampling problems by refining the framework obtained by \cite{AHLVXY22}. Recently, \cite{AHLVXY22} showed a parallel reduction from sampling to counting for \emph{discrete} distributions on the hypercube $\set{\pm 1}^n$, by combining a faithful discretization of stochastic localization and fast parallel sampling algorithms for continuous distributions.  For a discrete distribution $\mu$ over $\set{\pm 1}^n,$ their reduction involves $\log n$ iterations, each involving sampling from $\tau_w \mu \ast\mathcal{N}(0,c I) $ where $\tau_w \mu$ is the exponential tilt of $\mu$ by the vector $w \in \R^n,$ defined as:
\[ \tau_w \mu (x)\propto \exp(\dotprod{w, x})\, \mu(x)\,. \]
\cite{AHLVXY22} showed that for some appropriately chosen parameter $c=O(1)$, $\tau_w \mu \ast\mathcal{N}(0,c I) $ is a continuous and well-conditioned log-concave distribution for a wide class of discrete distributions $\mu$ of interest, i.e., those that are fractionally log-concave \cite[see][for a survey on fractional log-concavity]{AASV21}. In this way, they obtained a  parallel reduction to the problem of sampling from continuous and well-conditioned log-concave distributions.

The key technical challenge in their work is to control the propagation of errors resulting from the continuous sampler. Samples in an iteration become part of the external field $w$ at future steps. Assuming only the bound on $W_2$ guaranteed by \cite{SL19}, these errors can, in the worst case, be blown up by a factor of $\poly(n)$ in each iteration, resulting in a quasipolynomial blowup by the end. As a result, \cite{AHLVXY22} only manage to obtain $\log(n)^{O(1)}$ parallel time by using $n^{O(\log n)}$, that is \emph{quasipolynomially many}, processors (also known as a \Class{QuasiRNC} algorithm). For some specific distributions $\mu$, specifically strongly Rayleigh distributions~\citep{AHLVXY22}, they circumvent this shortcoming by establishing a property they call transport-stability for the distribution of interest, but several other notable distributions such as Eulerian tours and asymmetric determinantal point processes fall outside the reach of this trick. Here, by replacing the $W_2$ guarantee of \cite{SL19} with a TV distance guarantee, we entirely remove the need for transport-stability, turning the previous \Class{QuasiRNC} algorithms into \Class{RNC} algorithms.

%
%




Hence, our result implies an \Class{RNC}-time sampler for a fractionally log-concave distribution $\mu$ given access to an oracle which, given input $w\in \R^n,$ approximately computes the partition function of $\tau_w \mu$. This holds more generally for all $\mu$ whose tilts have constantly bounded covariance, i.e., $\cov{\tau_y \mu}\preceq O(1)\, I$, analogous to \cite{AHLVXY22}. 

The normalizing factor or partition function of $\tau_w \mu$ is $\sum_{x\in \{\pm\}^n} \exp(\dotprod{w, x})\,\mu(x)$. Viewed as a function of $w$, the partition function is also known as the Laplace transform of $\mu$. We denote the log of the partition function, a.k.a.\ the log-Laplace transform, by $\L_{\mu}{w} = \log \sum_{x\in \{\pm\}^n} \exp(\dotprod{w, x})\,\mu(x)$.
By an abuse of notation, we expand the definition of the Laplace transform to all vectors $w \in (\R \cup \{\pm \infty\})^n$ as follows.
Let $S$ be the set of coordinates $i$ where $w_i\in \set{\pm \infty}$, then: 
\[\L_{\mu}{w}  = \log \sum_{x\in \{\pm\}^n, \; \sign{x_S} = \sign{w_S}} \exp(\dotprod{w_{-S}, x_{-S}})\, \mu(x)\,.\]

\begin{definition}[Approximate oracle for the Laplace transform] \label{def:approximate laplace transform}
We say that the oracle $\mathcal{O}(\cdot)$ $\varepsilon$-approximately computes the log-Laplace transform at $\mu$ if on input $w$, $\mathcal{O}$ outputs $\exp(\hat{\mathcal{L}})$ s.t.\
\[ \abs{\hat{\mathcal{L}} - \L_{\mu}{w}}\leq \varepsilon\,. \]
\end{definition}
\begin{theorem}\label{thm:rnc discrete sampling}
	Suppose that a distribution $\mu$ on $\set{\pm 1}^n$  has $\cov{\tau_w \mu} \preceq O(1)\, I$ for all $w\in \R^n$, and we have an oracle for $O(\varepsilon/\sqrt{n})$-approximately computing the log-Laplace transform of $\mu$. Then we can sample from a distribution $\varepsilon$-close in total variation distance to $\mu$, in $\log(n/\varepsilon)^{O(1)}$ time using $(n/\varepsilon)^{O(1)}$ processors.
\end{theorem}

Thus, we improve upon \cite{AHLVXY22}'s reduction from sampling to counting in two ways:
\begin{itemize}
    \item We remove the assumption that the distribution needs to satisfy a transport inequality, which is only known to hold for strongly Rayleigh distributions and partition-constraint strongly Rayleigh distributions~\citep{AHLVXY22}. Under the weaker assumption of fractional log-concavity or bounded covariance under tilts, 
\cite{AHLVXY22} were only able to show a \Class{QuasiRNC} reduction from sampling to counting, i.e., their sampling algorithm uses $\approx n^{\log n}$ processors. 
    \item We only require an approximate counting oracle (see~\Cref{def:approximate laplace transform}) instead of the exact counting oracle required by \cite{AHLVXY22}. 
\end{itemize}
\Cref{thm:rnc discrete sampling} implies
the following corollary about asymmetric determinantal point processes (DPPs) and Eulerian tours \cite[see][for details and definitions]{AHLVXY22}.
\begin{corollary}\label{cor:app2}
	Suppose that $\mu$ is an asymmetric DPP on a ground set of size $n$ or the distribution of uniformly random Eulerian tours in a digraph of size $n$. Then, we can sample from a distribution $\varepsilon$-close in total variation distance to $\mu$ in time $\log(n/\varepsilon)^{O(1)}$ using $(n/\varepsilon)^{O(1)}$ processors.
\end{corollary}
Hence, we resolve \cite{AHSS21}'s question about designing an \Class{RNC} sampler for directed Eulerian tours.

Note that for the distributions studied in \cite{AHLVXY22}, counting can be done exactly via determinant computations, or in other words, there is exact access to the log-Laplace transform. But there are several non-exact approximate counting techniques in the literature that can be efficiently parallelized. A notable one is Barvinok's polynomial interpolation method \cite[see, e.g.,][]{BB21}. As an example of a distribution where Barvinok's method can be applied, consider a distribution $\mu$ on the hypercube $\set{\pm 1}^n$ defined by a polynomial Hamiltonian: $\mu(x)= \exp(p(x))$. \cite{BB21} showed that for quadratic and cubic polynomials $p$, assuming the coefficients of degree $2$ and $3$ terms are not too large (see \cite{BB21} for exact conditions), $\sum_{x\in \{\pm\}^n} \mu(x)$ can be approximately computed in quasipolynomial time. It can be observed that the approximation algorithm can be parallelized into a \Class{QuasiRNC} one since it simply involves computing $n^{\log n}$ separate quantities. We note that because the condition on $p$ does not involve the linear terms, we can also apply the same algorithm to $\tau_w \mu$, whose potential differs from $\mu$ only in the linear terms. In other words, Barvinok's method gives us the oracle in~\Cref{def:approximate laplace transform}. In the same paper, \cite{BB21} prove that the partition functions of these models are root-free in a sector, a condition known as sector-stability, which is known to imply fractional log-concavity~\citep{AASV21}. As a result, by plugging in Barvinok's approximate counting algorithm into our result, we obtain \Class{QuasiRNC} sampling algorithms, which at least in the case of cubic $p$ were not known before.
	\section{Preliminaries}

We let $\log$ denote the natural logarithm. For $x\in\R^d$, $\norm{x}$ denotes the usual Euclidean norm of $x$.

For two distributions $\rho$ and $\pi$, we use $\dTV{\rho, \pi}$ to denote their total variation distance defined as $\sup\set{\rho(E)-\pi(E)\given E~\text{is an event}}$.

A stronger notion of distance is the Kullback--Leibler (KL) divergence.

\begin{definition}[Kullback--Leibler divergence]
	For two probability densities $\rho, \pi$ we define
	\[ \DKL{\rho\river \pi}=\E_{\rho}\log(\rho/\pi)\,. \]
\end{definition}
We have the following relation between the KL divergence and TV distance, known as the Pinsker inequality.
\[ \dTV{\rho, \pi}\leq \sqrt{\frac{1}{2}\DKL{\rho\river \pi}}\,. \]

\subsection{Log-concave distributions}
Consider a density function $\pi: \R^d \to \R_{\geq 0}$ where $\pi(x)=\exp(-V(x))$. We call $V$ the potential function for $\pi$. Throughout the paper, we will assume that $V$ is twice continuously differentiable for simplicity of exposition.

\begin{definition}[Smoothness]
	For $\beta>0$, we say $\pi$ is $\beta$-smooth if the gradients of the potential are $\beta$-Lipschitz, that is
	\[ \norm{\nabla V(x)-\nabla V(y)}\leq \beta\, \norm{x-y}\,, \qquad \text{for all}~x,y\in\R^d\,. \]
	For twice differentiable $V$, this is equivalent to
	\[ -\beta I\preceq \nabla^2 V \preceq \beta I\,. \] 
\end{definition}

When $V$ is convex, we call $\pi$ a log-concave density. A strengthening of this condition is:
\begin{definition}[Strong log-concavity]
	For $\alpha>0$, we say $\pi$ is $\alpha$-strongly log-concave if 
\[  0\prec \alpha I \preceq \nabla^2 V\,.\]
\end{definition}


  \subsection{Log-Sobolev and transport-entropy inequalities}

\begin{definition}[Log-Sobolev inequality]
We say $\pi$ satisfies a log-Sobolev inequality (LSI) with constant $\alpha$ if for all smooth $f: \R^d \to \R$,
\[\Ent_{\pi} {f^2} \deq \E_{\pi}[f^2\log (f^2/ \E_{\pi}(f^2))] \leq \frac{2}{\alpha} \E_{\pi}[\norm{\nabla f}^2]\,. \]
\end{definition}

By the Bakry--\'Emery criterion~\citep{BE06}, if $\pi$ is $\alpha$-strongly log-concave then $ \pi $ satisfies LSI with constant $\alpha$.
The right-hand side of the above inequality can also be written as the relative Fisher information.

\begin{definition}[Relative Fisher information]
The relative Fisher information of $\rho$ w.r.t.\ $\pi$ is
\begin{equation}\label{ineq:log sobolev kl vs fisher}
  \FI{\rho \river \pi} = \E_{\rho}[\norm{\nabla \log (\rho/\pi)}^2]\,.
\end{equation}
\end{definition}

The LSI is equivalent to the following relation between KL divergence and Fisher information:
\[ \DKL{\rho \river \pi} \leq \frac{1}{2\alpha} \FI{\rho \river \pi} \qquad\text{for all probability measures}~\rho\,.\]
Indeed, take $f = \sqrt{\rho/\pi}$ in the above definition of the LSI\@.

\begin{definition}[Wasserstein distance]
We denote by $W_2$ the Wasserstein distance between $ \rho$ and $\pi$, which is defined as
\[W_2^2 (\rho, \pi) =\inf\bigl\{\E_{(X, Y)\sim \Pi}[\norm{X-Y}^2 ]\given \Pi \text{ is a coupling of }\rho, \pi\bigr\}\,, \]
where the infimum is over coupling distributions $\Pi$ of $(X,Y)$ such that $X\sim \rho, Y\sim \pi$.
\end{definition}

The log-Sobolev inequality implies the following transport-entropy inequality, known as Talagrand's $T_2$ inequality~\citep{OTTO2000361}:
\begin{equation} \label{ineq:talagrand}
  \frac{\alpha}{2}\, W_2^2 (\rho, \pi) \leq  \DKL{\rho \river \pi}\,.
\end{equation}
	\section{Parallel sampling guarantees}

In this section, we formally state our main parallel sampling guarantees.

\subsection{LMC}\label{sec:lmc}

We state the formal version of~\Cref{thm:informal-main} for LMC as~\Cref{thm:sample-main}.
Our assumption throughout is that the score function $s$ is a pointwise accurate estimate of $\nabla V$:
\begin{assumption}
The score function $s:\R^d \to \R$ satisfies $\norm{s(x) - \nabla V(x)} \leq \delta$ for all $x\in \R^d.$
\end{assumption}

\begin{theorem} \label{thm:sample-main}
Suppose that $V$ is $\beta$-smooth and $\pi$ satisfies a log-Sobolev inequality with constant $\alpha$, and the score function $s$ is $\delta$-accurate. Let $\kappa \deq \beta/\alpha$. 
    Suppose
    \begin{equation}
    \begin{aligned}
     \beta h &\leq 1/10\,,&\qquad \delta&\leq 2\sqrt\alpha\varepsilon\,, &\qquad
     M&\geq 7\max \{\kappa d/\varepsilon^2, \kappa^2\}\,, \\[0.5em]
     K &\geq 2+ \log M\,, &&&\qquad Nh &\geq \alpha^{-1}\log \frac{2 \DKL{\mu_{0} \river\pi}}{\varepsilon^2}\,.
     \end{aligned}
    \end{equation}
Then, the output distribution $\mu_{Nh}$ of~\Cref{alg:main} satisfies 
\[\max\Bigl\{
\frac{\sqrt{\alpha}}{2}\, W_2(\mu_{Nh}, \pi ), \dTV{\mu_{Nh}, \pi} \Bigr\}\leq \sqrt{\frac{\DKL{\mu_{Nh} \river \pi}}{2}}\leq \varepsilon. \] 
 \end{theorem}

To make the guarantee more explicit, we can combine it with the following well-known initialization bound, see, e.g.,~\citet[\S 3.2]{Dwietal19MALA}.
\begin{corollary}\label{cor:lmc_guarantee}
 Suppose that $\pi = \exp(-V)$ with $0 \prec \alpha I \preceq \nabla^2 V \preceq \beta I$, and let $\kappa \deq \beta/\alpha$. Let $x^\star$ be the minimizer of $V$. Then, for $\mu_0 = \mc N(x^\star, \beta^{-1} I)$, it holds that $ \DKL{\mu_0\river \pi} \leq 
 \frac{d}{2}\log \kappa$. 

Consequently, setting
 \[ h = \frac{1}{10 \beta}\,, \quad \delta = 2\sqrt\alpha\varepsilon\,, \quad M = 7\max\bigl\{\frac{\kappa d}{\varepsilon^2}, \kappa^2\bigr\}\,, \quad K = 3\log M\,,\quad N =  10 \kappa \log  \frac{d \log \kappa }{\varepsilon^2}\,, \]
 then~\Cref{alg:main} initialized at $\mu_0$ outputs $\mu_{Nh}$ satisfying
\[\max\bigl\{
\frac{\sqrt{\alpha}}{2} \,W_2(\mu_{Nh}, \pi ), \dTV{\mu_{Nh}, \pi} \bigr\}\leq \sqrt{\frac{\DKL{\mu_{Nh} \river \pi}}{2}}\leq \varepsilon\,.\]
Also,~\Cref{alg:main} uses a total of $KN = \widetilde O(\kappa \log^2(d/\varepsilon^2))$ parallel rounds and $M$ $\delta$-approximate gradient evaluations in each round.
\end{corollary}

The proofs for this section are given in \S\ref{app:lmc}.


\subsection{ULMC}\label{sec:ulmc}

In this section, we design a parallel sampler based on underdamped Langevin Monte Carlo (ULMC), also called \emph{kinetic Langevin}, which has similar parallel iteration complexity as LMC but requires less total work.
Since there are difficulties applying the interpolation method without higher-order smoothness assumptions (see the discussion in~\cite{Maetal21NesterovMCMC, Zhangetal23ULMC}), we will use a different proof technique based on Girsanov's theorem, as in~\cite{AltChe23warm, Zhangetal23ULMC}.
Note that since we seek TV guarantees, we cannot apply the coupling arguments of~\citet{Chengetal18ULMC, DalRio20Kinetic}.

\subsubsection{Algorithm}\label{sec:ulmc_alg}

In continuous time, the underdamped Langevin diffusion is the coupled system of SDEs
\begin{align*}
    dX_t
    &= P_t \, dt\,, \\
    dP_t
    &= -\nabla V(X_t)\, dt - \gamma P_t \, dt + \sqrt{2\gamma}\, dB_t\,,
\end{align*}
where $\gamma > 0$ is the friction parameter. Throughout, we will simply set $\gamma = \sqrt{8\beta}$, where $\beta$ is the smoothness parameter.

The idea for developing a parallel sampler is similar as before: we parallelize Picard iteration.
However, in order to eventually apply Girsanov's theorem to analyze the algorithm, the discretization must be chosen so that $dX_t = P_t \, dt$ is preserved. Hence, we will use the exponential Euler integrator.

We use the following notation: $\tau(t)$ is the largest multiple of $h/M$ which is less than $t$, i.e., $\tau(t) = \lfloor t/\frac{h}{M} \rfloor \, \frac{h}{M}$.
We define a sequence of processes $(X^{(0)}, P^{(0)})$, $(X^{(1)}, P^{(1)})$, etc., so that
\begin{align*}
    dX_t^{(k+1)} &= P_t^{(k+1)}\,dt \,, \\
    dP_t^{(k+1)} &= -\nabla V(X_{\tau(t)}^{(k)})\,dt - \gamma P_t^{(k+1)}\, dt + \sqrt{2\gamma}\, dB_t\,.
\end{align*}
This is a linear SDE\@, so it can be integrated exactly, yielding
\begin{align}
    X_{nh+(m+1)h/M}^{(k+1)}
    &= X_{nh+mh/M}^{(k+1)} + \frac{1-\exp(-\gamma h/M)}{\gamma} \, P_{nh+mh/M}^{(k+1)} \nonumber \\
    &\qquad{} - \frac{h/M - (1-\exp(-\gamma h/M))/\gamma}{\gamma}\,\nabla V(X_{nh+mh/M}^{(k)}) + \xi^X\,, \label{eq:ulmc_X} \\
    P_{nh+(m+1)h/M}^{(k+1)}
    &= \exp(-\gamma h/M)\, P_{nh+mh/M}^{(k+1)} - \frac{1-\exp(-\gamma h/M)}{\gamma}\,\nabla V(X_{nh+mh/M}^{(k)}) + \xi^P\,, \label{eq:ulmc_P}
\end{align}
where $(\xi^X, \xi^P)$ is a correlated Gaussian vector in $\R^d\times\R^d$ with law $\mc N(0,\Sigma)$, where
\begin{align}\label{eq:correlated_gaussian}
    \Sigma = \begin{bmatrix}
        \frac{2}{\gamma} \,[\frac{h}{M} - \frac{2}{\gamma} \,(1-\exp(-\gamma h/M)) + \frac{1}{2\gamma} \,(1-\exp(-2\gamma h/M))] & * \\[0.5em]
        \frac{1}{\gamma} \,(1-2\exp(-\gamma h/M) +\exp(-2\gamma h/M)) & 1-\exp(-2\gamma h/M)
    \end{bmatrix}\,,
\end{align}
and the upper-left entry marked $*$ is determined by symmetry.

Note that each processor $m=1,\dotsc,M$ can independently generate a correlated Gaussian vector according to the above law and store it. Then, the updates for the above discretization can be computed quickly in parallel.
We summarize the algorithm below as \Cref{alg:ulmc}.

\begin{algorithm}[t]
\textbf{Input}: $(X_0,P_0)\sim \mu_0$, approximate score function $s: \R^d \to \R^d$ ($s\approx \nabla V$) \;

\For{$n=0, \dots, N-1$}{
    \For{$m=0, \dots, M$ in parallel}{
        $(X_{nh+mh/M}^{(0)}, P_{nh+mh/M}^{(0)}) \gets (X_{nh}, P_{nh})$\;
        
        Sample correlated Gaussian vectors according to~\cref{eq:correlated_gaussian}
    }
    \For{$k=0, \dots, K-1$}{
        \For{$m=0, \dots, M$ in parallel}{
        Compute $(X_{nh+mh/M}^{(k+1)}, P_{nh+mh/M}^{(k+1)})$ using~\cref{eq:ulmc_X} and~\cref{eq:ulmc_P}, replacing $\nabla V$ with $s$
        }    
    }
    $(X_{(n+1)h}, P_{(n+1)h}) \gets (X_{nh+h}^{(K)}, P_{nh+h}^{(K)})$
    \caption{Parallelized underdamped Langevin dynamics \label{alg:ulmc}}	
}
\end{algorithm}

\subsubsection{Analysis}\label{sec:ulmc_analysis}

We now give our guarantees for~\Cref{alg:ulmc}. Compared to~\Cref{thm:sample-main}, it improves the number of processors by roughly a factor of $\sqrt{\kappa d}/\varepsilon$.
Although it is stated for strongly log-concave measures for simplicity, similarly to \S\ref{sec:lmc}, the discretization guarantees only require $\pi$ to satisfy a log-Sobolev inequality and smoothness; see~\Cref{thm:ulmc_discretization} for a more precise statement. The proof is given in \S\ref{app:ulmc}.

\begin{theorem}\label{thm:ulmc_final}
    Assume that $V$ is $\alpha$-strongly convex and $\beta$-smooth; let $\kappa \deq \beta/\alpha$.
    Assume that $V$ is minimized at $x^\star$.
    Consider~\Cref{alg:ulmc} initialized at $\mu_0 = \mc N(x^\star, \beta^{-1} I) \otimes \mc N(0, I)$ and with
    \begin{align*}
        h = \Theta\bigl(1/\sqrt\beta\bigr)\,,\;\; \delta \le \widetilde O\bigl(\frac{\sqrt\alpha\varepsilon}{\sqrt{\log d}}\bigr)\,,\;\; M = \widetilde \Theta\bigl(\frac{\sqrt{\kappa d}}{\varepsilon}\bigr)\,,\;\; K = \Theta\bigl(\log \frac{\kappa d}{\varepsilon^2}\bigr)\,,\;\; N = \widetilde\Theta\bigl(\kappa \log \frac{d}{\varepsilon^2}\bigr)\,.
    \end{align*}
    Then, the law of the output of~\Cref{alg:ulmc} is $\varepsilon$-close in total variation distance to $\pi$.
    The algorithm uses a total of $KN = \widetilde\Theta(\kappa \log^2(d/\varepsilon^2))$ parallel rounds and $M$ $\delta$-approximate gradient evaluations in each round.
\end{theorem}
	\section{Implications for sampling from discrete distributions}

In this section, we prove~\Cref{thm:rnc discrete sampling}.
For simplicity, we only state our parallel guarantees using parallel LMC, for which the initialization is more straightforward, but it is easy to combine the results of this section with parallel ULMC as well.
For concreteness, we restate \cite{AHLVXY22}'s sampling-to-counting reduction.
Then,~\Cref{thm:rnc discrete sampling} is a consequence of~\citet[Lemma 7]{AHLVXY22}, our fast parallel sampler with TV guarantee, and a modified version of~\citet[Proposition 27]{AHLVXY22}. We include the proofs for completeness in \S\ref{app:discrete}. 

\begin{algorithm}[t]
    Initialize $w_0\gets 0$ \\
	\For{$i = 0,\dotsc,T-1$}{
	$x_{i+1}\gets$ (approximate) sample from $\tau_{w_i}\mu*\Normal{0, cI}$ \\
	$w_{i+1}\gets w_i+x_{i+1}/c$
 }
\textbf{return} $\sign{w_T}\in \set{\pm 1}^n$
\caption{Framework for discrete sampling via continuous sampling \label{alg:sampling to counting reduction}}	
\end{algorithm}

We give the overall algorithm as~\Cref{alg:sampling to counting reduction}.
The following lemma shows that the step of sampling from distributions of the form $\tau_w \mu * \Normal{0, cI}$ is a well-conditioned log-concave sampling problem, and moreover, that the score can be approximated quickly in parallel.

\begin{lemma}[\cite{AHLVXY22}]\label{lem:computation}
    Let $\nu = \tau_w \mu \ast \Normal{0,cI}$. Then, $\nu \propto \exp(-V)$ with \[ -\nabla V(y) = \frac{\mean{\tau_{y/c+w} \mu} }{c} - \frac{y}{c} = \frac{1}{c}\, \frac{\sum_{x\in \{\pm\}^n} x \exp(\dotprod{y/c+w,x})\, \mu(x) }{\sum_{x\in\{\pm\}^n} \exp(\dotprod{y/c+w,x})\, \mu(x) }  -\frac{y}{c} \] 
    and
    \[\nabla^2 V(y) = -\frac{\cov{\tau_{y/c + w}\mu}  }{c^2} + \frac{I}{c}\,.\]
    If $\cov{\tau_y\mu} \preceq \frac{c}{2} I$ for all $y\in\R^n$, then $\nu$ is well-conditioned strongly log-concave with condition number $\kappa = O(1)$, i.e., for all $y \in\R^n$:
    \[ \frac{1}{2c}\,I \preceq \nabla^2 V(y) \preceq \frac{1}{c}\,I\,. \]

   Furthermore, a $\delta$-approximate score function $s$ for $\nabla V$ can be computed in $O(1)$ parallel iterations using $ n$ machines, each making $O(1)$ calls to an $\varepsilon = O(\delta\sqrt{c/n})$-approximate oracle for the Laplace transform of $\mu.$ 
    
\end{lemma}

The next lemma states that if the samples from the continuous densities $\tau_w \mu * \mc N(0, cI)$ are accurate, then the output of~\Cref{alg:sampling to counting reduction} outputs an approximate sample from $\mu$.

\begin{lemma}[{\citet[Lemma 7]{AHLVXY22}}]\label{lem:tv-close}
    If the continuous samples are exact in~\Cref{alg:sampling to counting reduction}, then for $T=\Omega(c\log(n/\varepsilon))$, the distribution of $ c w_T/T$ is $\mu\ast \Normal{0,\frac{c}{T} I}$ and output of the algorithm is $\varepsilon$-close in total variation distance to $\mu$. 
\end{lemma}

These results, together with an initialization bound (see~\Cref{lem:initialization bound for discrete convolve with Gaussian}), then yield the proof of~\Cref{thm:rnc discrete sampling}.
Details are given in \S\ref{app:discrete}.

    

    \section*{Acknowledgements}

    SC acknowledges the support of the Eric and Wendy Schmidt Fund at the Institute for Advanced Study.

    \newpage
    \appendix
    \section{Proofs for LMC}\label{app:lmc}

In this section, we give the proofs for \S\ref{sec:lmc}.
Let $\mu_{nh} \deq \law(X_{nh})$.
We first need the following recursive bound, which shows that the error decays exponentially fast in the parallel refinement.

\begin{lemma}\label{lem:parallel refinement}
    Suppose that $V$ is $\beta$-smooth, and that the score function $s$ is $\delta$-accurate.
    Assume that $\beta h\le 1/10$ and that $\pi$ satisfies Talagrand's $T_2$ inequality with constant $\alpha$.
    Then,
    \begin{align*}
        &\max_{m=1,\dotsc,M} \E[\norm{X_{nh+mh/M}^{(K)} - X_{nh+mh/M}^{(K-1)}}^2] \\
        &\qquad \le 34\exp(-3.5K)\,\Bigl(1.4dh + \frac{8\beta^2 h^2}{\alpha}\DKL{\mu_{nh}\river \pi}\Bigr) + 8.2\delta^2 h^2\,.
    \end{align*}
\end{lemma}
\begin{proof}
    Let
    \begin{align*}
        \mc E_k \deq \max_{m=1,\dotsc,M} \E[\norm{X_{nh+mh/M}^{(k)} - X_{nh+mh/M}^{(k-1)}}^2]\,.
    \end{align*}
    For any $m = 1,\dotsc,M$,
    \begin{align*}
        &\E[\norm{X_{nh+mh/M}^{(k+1)} - X_{nh+mh/M}^{(k)}}^2]
        = \Ex\Bigl[\Bigl\Vert \frac{h}{M} \sum_{m'=1}^{m-1} \bigl(s(X_{nh+m'h/M}^{(k)}) - s(X_{nh+m'h/M}^{(k-1)})\bigr) \Bigr\rVert^2\Bigr] \\
        &\qquad \le \frac{h^2 m}{M^2} \sum_{m'=1}^{m-1} \E[\norm{s(X_{nh+m'h/M}^{(k)}) - s(X_{nh+m'h/M}^{(k-1)})}^2] \\
        &\qquad \le 3h^2 \max_{m'=1,\dotsc,m} \E[\norm{\nabla V(X_{nh+m'h/M}^{(k)}) - \nabla V(X_{nh+m'h/M}^{(k-1)})}^2] +6\delta^2 h^2 \\
        &\qquad \le 3\beta^2 h^2\,\mc E_k + 6\delta^2 h^2
    \end{align*}
    and hence $\mc E_{k+1} \le 3\beta^2 h^2\, \mc E_k + 6\delta^2 h^2$.
    Also,
    \begin{align*}
        \E[\norm{X_{nh+mh/M}^{(1)} - X_{nh}}^2]
        &= \frac{h^2 m^2}{M^2} \E[\norm{s(X_{nh})}^2] + \frac{dhm}{M} \\
        &\le 2\delta^2 h^2 +2h^2 \E[\norm{\nabla V(X_{nh})}^2] + dh
    \end{align*}
    and thus $\mc E_1$ is bounded by the right-hand side above.
    Iterating the recursion and using $\beta h\le 10$,
    \begin{align*}
        \mc E_K
        &\le \exp(-3.5\,(K-1))\,\mc E_1 + 6.2\delta^2 h^2 \\
        &\le \exp(-3.5\,(K-1))\,\{2\delta^2 h^2 +2h^2 \E[\norm{\nabla V(X_{nh})}^2] + dh\} + 6.2\delta^2 h^2\,.
    \end{align*}
    Also, by~\citet[Lemma 10]{VW19},
    \begin{align*}
        \E[\norm{\nabla V(X_{nh})}^2]
        &\le 2\beta d + \frac{4\beta^2}{\alpha}\DKL{\mu_{nh} \river \pi}\,.
    \end{align*}
    Substituting this in and using $\beta h\le 1/10$ yields the result.
\end{proof}

\begin{proof}[Proof of~\Cref{thm:sample-main}]
We will use the interpolation method.
Let $X_{nh+mh/M} = X_{nh+mh/M}^{(K)}.$ It is easy to see that
\[X_{nh + (m+1)h/M} = X_{nh+mh/M} - \frac{h}{M}\, s (X_{nh+mh/M}^{(K-1)}) +\sqrt{2}\, (B_{nh + (m+1)h/M} -B_{nh+mh/M})\,. \]
Let $X$ denote the interpolation of $X^{(K)}$, i.e., for $ t\in [nh+mh/M, nh +(m+1)h/M],$  let
\begin{align*}
    X_t  = X_{nh+mh/M} - (t-nh-mh/M)\, s (X_{nh+mh/M}^{(K-1)}) +\sqrt{2}\, (B_t -B_{nh+mh/M})\,.
\end{align*}
Note that $s(X_{nh+mh/M}^{(K-1)})$ is a constant vector field given $ X_{nh+mh/M}^{(K-1)}.$ Let $\mu_t$ be the law of $ X_t.$
The same argument as in~\citet[Proof of Lemma 3]{VW19} yields the differential inequality
\begin{equation} \label{ineq:time derivative of KL}
    \begin{split}
        \partial_t \DKL{\mu_t \river\pi} &= -\FI {\mu_t \river \pi} + \Ex\Bigl\langle \nabla V(X_t) - s(X_{nh+mh/M}^{(K-1)}), \nabla \log\frac{ \mu_t(X_t) }{\pi(X_t)}\Bigr\rangle\\
    &\leq -\frac{3}{4} \FI {\mu_t \river \pi} +  \E[\|\nabla V(X_t) - s(X_{nh+mh/M}^{(K-1)}) \|^2]
    \end{split}
\end{equation}
where we used $\langle a, b \rangle \leq \norm{a}^2 + \frac{1}{4}\, \|b\|^2$ and $\E[\norm{ \nabla \log\frac{ \mu_t(X_t) }{\pi(X_t)}}^2] = \FI{\mu_t \river \pi}$.
Next, we bound
\begin{equation} \label{ineq:discretization bound}
\begin{split}
    &\E[\|\nabla V(X_t) - s(X_{nh+ mh/M}^{(K-1)}) \|^2] \\
    &\qquad \leq 2 \E[\|\nabla V(X_t) - \nabla V(X_{nh+ mh/M}^{(K-1)}) \|^2  + \norm{\nabla V(X_{nh+ mh/M}^{(K-1)})   - s(X_{nh+ mh/M}^{(K-1)})}^2]   \\
    &\qquad \leq 2\E[\|\nabla V(X_t) - \nabla V(X_{nh+ mh/M}^{(K-1)}) \|^2] +2 \delta^2\\
    &\qquad \leq 2\beta^2 \E[\|X_t - X_{nh+ mh/M}^{(K-1)} \|^2] + 2\delta^2\,.
\end{split}    
\end{equation}
Moreover,
\begin{align}\label{eq:disc_1}
    \E[\norm{X_t - X_{nh+mh/M}^{(K-1)}}^2]
    &\le 2\E[\norm{X_t - X_{nh+mh/M}}^2] + 2\E[\norm{X_{nh+mh/M}^{(K)} - X_{nh+mh/M}^{(K-1)}}^2]\,.
\end{align}
The first term above is
\begin{align*}
    &\E[\norm{X_t - X_{nh+mh/M}}^2]
    = (t-nh-mh/M)^2 \E[\norm{s(X_{nh+mh/M}^{(K-1)})}^2] + d\,(t-nh-mh/M) \\
    &\qquad \le \frac{2h^2}{M^2} \E[\norm{\nabla V(X_{nh+mh/M}^{(K-1)})}^2] + \frac{2\delta^2 h^2}{M^2} + \frac{dh}{M} \\
    &\qquad \le \frac{4\beta^2 h^2}{M^2} \E[\norm{X_t - X_{nh+mh/M}^{(K-1)}}^2] + \frac{4h^2}{M^2} \E[\norm{\nabla V(X_t)}^2] + \frac{2\delta^2 h^2}{M^2} + \frac{dh}{M}\,.
\end{align*}
Substituting this into~\cref{eq:disc_1} and using $\beta h\le 1/10$ yields
\begin{align*}
    \E[\norm{X_t - X_{nh+mh/M}^{(K-1)}}^2]
    &\le \frac{4.4 h^2}{M^2} \E[\norm{\nabla V(X_t)}^2] + \frac{2.2\delta^2 h^2}{M^2} + \frac{1.1dh}{M} \\
    &\qquad{} + 2.2\E[\norm{X_{nh+mh/M}^{(K)} - X_{nh+mh/M}^{(K-1)}}^2]\,.
\end{align*}
Now,~\citet[Lemma 16]{chewi2021analysis} yields
\begin{align*}
    \E[\norm{\nabla V(X_t)}^2]
    &\le \FI{\mu_t \river \pi} + 2\beta d\,.
\end{align*}
For the last term, we can apply \Cref{lem:parallel refinement}.

  Substituting everything into \cref{ineq:time derivative of KL} and cleaning up the terms yields
\begin{align*}
\partial_t \DKL{\mu_t \river\pi}
&\le -0.66\FI{\mu_t\river \pi} + 2.5\delta^2 \\
&\qquad{}+ 2\beta^2 \,\Bigl[\frac{2dh}{M} + 75\exp(-3.5K)\,\Bigl(1.4dh + \frac{8\beta^2 h^2}{\alpha}\DKL{\mu_{nh}\river\pi}\Bigr)\Bigr]\,.
\end{align*}
Assuming that $K \ge 1.3+0.3 \log M$, and using the LSI\@,
\begin{align*}
    \partial_t \DKL{\mu_t \river\pi}
    &\le -1.3\alpha \DKL{\mu_t\river \pi} + 2.5\delta^2 + \frac{6.8\beta^2 dh}{M} + \frac{16\beta^4 h^2}{\alpha M} \DKL{\mu_{nh}\river\pi}\,.
\end{align*}
Integrating this inequality,
\begin{align*}
    \DKL{\mu_{(n+1)h} \river \pi}
    &\le \Bigl[\exp(-1.3\alpha h) + \frac{16\beta^4 h^3}{\alpha M}\Bigr] \DKL{\mu_{nh} \river \pi} + 2.5\delta^2 h + \frac{6.8\beta^2 dh^2}{M}\,.
\end{align*}
Provided $M \ge 6.4\kappa^2$, then $\exp(-1.3\alpha h) + \frac{16\beta^4}{h^3}{\alpha M} \le \exp(-\alpha h)$.
Iterating,
\begin{align*}
    \DKL{\mu_{Nh}\river \pi}
    &\le \exp(-\alpha Nh) \DKL{\mu_0\river \pi} + \frac{2.8\delta^2}{\alpha} + \frac{7.5\beta^2 dh}{\alpha M}\,.
\end{align*}
Thus we obtain the guarantee in KL divergence. The guarantees in TV and $W_2$ distance follow from Pinsker's and Talagrand's inequality respectively.
\end{proof}

\section{Proofs for ULMC}\label{app:ulmc}

We turn towards the analysis of~\Cref{alg:ulmc}.
We start by bounding the discretization error between the algorithm and the continuous-time process using Girsanov's theorem.
Throughout, let $\mu_{Nh}$ denote the law of the output of the algorithm, and let $\pi_t$ denote the marginal law of the continuous-time Langevin diffusion at time $t$ started from $\mu_0$.

First, we need a lemma.

\begin{lemma}\label{lem:ulmc_moment_bds}
    Let $(X_t, P_t)_{t\ge 0}$ denote the continuous-time underdamped Langevin diffusion, started at $(X_0, P_0) \sim \mu_0$.
    Assume that $V$ is $\beta$-smooth, and that $\pi^X \propto \exp(-V)$ satisfies Talagrand's $T_2$ inequality with constant $\alpha$.
    Let $\pi = \pi^X \otimes \mc N(0, I)$.
    Then,
    \begin{align*}
        \Ex[\norm{\nabla V(X_t)}^2] \le 2\beta d +\frac{4\beta^2}{\alpha} \DKL{\mu_0\river \pi}\,, \qquad \Ex[\norm{P_t}^2] \le 2d +\DKL{\mu_0\river \pi}\,.
    \end{align*}
\end{lemma}
\begin{proof}
    For the first bound, we use a similar proof as~\citet[Lemma 10]{VW19}.
    Namely, by Lipschitzness of $\nabla V$, the transport inequality, and the data-processing inequality,
    \begin{align*}
        \Ex[\norm{\nabla V(X_t)}^2]
        &\le 2\,\Ex_{\pi^X}[\norm{\nabla V}^2] + 2\beta^2 \, W_2^2(\law(X_t), \pi^X)
        \le 2\beta d + \frac{4\beta^2}{\alpha} \DKL{\law(X_t) \river \pi^X} \\
        &\le 2\beta d + \frac{4\beta^2}{\alpha} \DKL{\law(X_t, P_t) \river \pi}
        \le 2\beta d + \frac{4\beta^2}{\alpha} \DKL{\mu_0 \river \pi}\,.
    \end{align*}
    Similarly,
    \begin{align*}
        \Ex[\norm{P_t}^2]
        &\le 2\,\Ex_{\mc N(0, I)}[\norm{\cdot}^2] + 2\,W_2^2(\law(P_t), \mc N(0,I))
        \le 2d + 4\DKL{\mu_0 \river \pi}\,.
    \end{align*}
    This completes the proof.
\end{proof}

We now state and prove our main discretization bound.

\begin{theorem}\label{thm:ulmc_discretization}
    Suppose that $V$ is $\beta$-smooth and that $\pi^X\propto\exp(-V)$ satisfies Talagrand's $T_2$ inequality with constant $\alpha$. Let $\kappa \deq \beta/\alpha$.
    Assume that the parallel depth satisfies $K \gtrsim \log M$ (for a sufficiently large implied constant) and that $h \lesssim 1/\sqrt\beta$ (for a sufficiently small implied constant).
    Then, it holds that
    \begin{align*}
        \DKL{\pi_T \river \mu_T}
        &\lesssim \frac{T}{\sqrt\beta}\,\Bigl(\delta^2 + \frac{\beta^2 dh^2}{M^2} + \frac{\beta^2 h^2}{M^2}\,\bigl(1+\frac{\kappa}{M^2}\bigr)\DKL{\mu_0\river \pi}\Bigr)\,.
    \end{align*}
\end{theorem}
\begin{proof}
    Let $\mb P$ denote the Wiener measure on $[0,T]$, under which $(B_t)_{t\in [0,T]}$ is a standard Brownian motion.
    Using this Brownian motion, we define the algorithm process, i.e.,
    \begin{align*}
        dX_t^{(k+1)}
        &= P_t^{(k+1)}\, dt\,, \\
        dP_t^{(k+1)}
        &= -s(X_{\tau(t)}^{(k)})\, dt -\gamma P_t^{(k+1)}\, dt + \sqrt{2\gamma}\, dB_t\,.
    \end{align*}
    We also drop the superscripts for parallel depth $K$, i.e., $(X_t^{(K)}, P_t^{(K)}) = (X_t, P_t)$.
    We now write
    \begin{align*}
        dP_t
        &= -\nabla V(X_t)\, dt - \gamma P_t \, dt + \sqrt{2\gamma} \, d\tilde B_t
    \end{align*}
    where $d\tilde B_t = dB_t - \frac{1}{\sqrt{2\gamma}}\,(s(X_{\tau(t)}^{(K-1)}) - \nabla V(X_t)) \, dt$.
    By Girsanov's theorem~\citep[see][\S 5.6]{LeG16StocCalc}, if we define the path measure $\mb Q$ via
    \begin{align}\label{eq:girsanov}
        \frac{d\mb Q}{d\mb P}
        &= \exp\Bigl( \frac{1}{\sqrt{2\gamma}} \int_0^T \langle s(X_{\tau(t)}^{(K-1)}) - \nabla V(X_t), dB_t\rangle - \frac{1}{8\gamma} \int_0^T \norm{s(X_{\tau(t)}^{(K-1)})-\nabla V(X_t)}^2 \, dt \Bigr)\,,
    \end{align}
    then under $\mb Q$ the process $\tilde B$ is a standard Brownian motion. It follows readily that under $\mb Q$, the process $(X, P)$ is the continuous-time underdamped Langevin diffusion.
    By the data-processing inequality and~\cref{eq:girsanov},
    \begin{align}
        &\DKL{\pi_T \river \mu_T}
        \le \DKL{\mb Q \river \mb P}
        = \Ex_{\mb Q} \log \frac{d\mb Q}{d\mb P} \nonumber \\
        &\qquad = \Ex_{\mb Q}\Bigl[\frac{1}{\sqrt{2\gamma}} \int_0^T \langle s(X_{\tau(t)}^{(K-1)}) - \nabla V(X_t), dB_t\rangle -\frac{1}{8\gamma}\int_0^T \norm{s(X_{\tau(t)}^{(K-1)})-\nabla V(X_t)}^2 \, dt\Bigr] \nonumber \\
        &\qquad = \Ex_{\mb Q}\Bigl[\frac{1}{\sqrt{2\gamma}} \int_0^T \langle s(X_{\tau(t)}^{(K-1)}) - \nabla V(X_t), d\tilde B_t\rangle +\frac{1}{8\gamma}\int_0^T \norm{s(X_{\tau(t)}^{(K-1)})-\nabla V(X_t)}^2 \, dt\Bigr] \nonumber \\
        &\qquad = \frac{1}{8\gamma}\,\Ex_{\mb Q} \int_0^T \norm{s(X_{\tau(t)}^{(K-1)})-\nabla V(X_t)}^2 \, dt\,.\label{eq:ulmcpf6}
    \end{align}
    From now on, all expectations are taken under $\mb Q$ and we drop the subscript $\mb Q$ from the notation.
    We focus on $t$ lying in the interval $[nh, (n+1)h]$.
    
    Of course, using the fact that we have $\delta$-accurate gradient evaluations,
    \begin{align}
        \Ex[\norm{s(X_{\tau(t)}^{(K-1)})-\nabla V(X_t)}^2]
        &\lesssim \delta^2 + \Ex[\norm{\nabla V(X_{\tau(t)}^{(K-1)})-\nabla V(X_t)}^2] \nonumber\\
        &\le \delta^2 + \beta^2 \,\Ex[\norm{X_{\tau(t)}^{(K-1)}-X_t}^2]\,. \label{eq:ulmcpf4}
    \end{align}
    We split this into two terms:
    \begin{align}\label{eq:ulmcpf5}
        \Ex[\norm{X_{\tau(t)}^{(K-1)}-X_t}^2]
        &\lesssim \Ex[\norm{X_t - X_{\tau(t)}}^2] + \Ex[\norm{X_{\tau(t)} - X_{\tau(t)}^{(K-1)}}^2]\,.
    \end{align}
    We begin with the recursive term (the second one).

    For any $k=1,\dotsc,K$, let
    \begin{align*}
        \mc E_k
        &\deq \max_{m=1,\dotsc,M} \Ex[\norm{X_{nh+mh/M}^{(k)} - X_{nh+mh/M}^{(k-1)}}^2]\,.
    \end{align*}
    To bound this quantity, we start with
    \begin{align}
        \Ex[\norm{X_{nh+mh/M}^{(k)} - X_{nh+mh/M}^{(k-1)}}^2]
        &= \Ex\Bigl[\Bigl\lVert \int_{nh}^{nh+mh/M} (P_t^{(k)} - P_t^{(k-1)}) \, dt \Bigr\rVert^2\Bigr] \nonumber \\
        &\le h\int_{nh}^{nh+mh/M} \Ex[\norm{P_t^{(k)} - P_t^{(k-1)}}^2] \, dt\,. \label{eq:ulmcpf1}
    \end{align}
    Next,
    \begin{align*}
        \Ex[\norm{P_t^{(k)} - P_t^{(k-1)}}^2]
        &= \Ex\Bigl[\Bigl\lVert \int_{nh}^t \{-(s(X_{\tau(s)}^{(k-1)}) - s(X_{\tau(s)}^{(k-2)})) - \gamma\,(P_s^{(k)} - P_s^{(k-1)})\} \, ds \Bigr\rVert^2\Bigr] \\
        &\lesssim h\int_{nh}^t \Ex[\norm{s(X_{\tau(s)}^{(k-1)}) - s(X_{\tau(s)}^{(k-2)})}^2 + \gamma^2 \,\norm{P_s^{(k)} - P_s^{(k-1)}}^2] \, ds\,.
    \end{align*}
    By Gr\"onwall's inequality,
    \begin{align*}
        \Ex[\norm{P_t^{(k)} - P_t^{(k-1)}}^2]
        &\lesssim h\exp(O(\gamma^2 h^2)) \int_{nh}^t\Ex[\norm{s(X_{\tau(s)}^{(k-1)}) - s(X_{\tau(s)}^{(k-2)})}^2] \, ds\,.
    \end{align*}
    Recall that $\gamma^2 \asymp \beta$. We assume throughout that $h \lesssim 1/\sqrt\beta$ for a sufficiently small implied constant, so that $\gamma^2 h^2 \lesssim 1$.
    Therefore,
    \begin{align*}
        \Ex[\norm{P_t^{(k)} - P_t^{(k-1)}}^2]
        &\lesssim h \int_{nh}^t\Ex[\norm{s(X_{\tau(s)}^{(k-1)}) - s(X_{\tau(s)}^{(k-2)})}^2] \, ds \\
        &\lesssim \delta^2 h^2 + h \int_{nh}^t\Ex[\norm{\nabla V(X_{\tau(s)}^{(k-1)}) - \nabla V(X_{\tau(s)}^{(k-2)})}^2] \, ds \\
        &\lesssim \delta^2 h^2 + \beta^2 h \int_{nh}^t\Ex[\norm{X_{\tau(s)}^{(k-1)} - X_{\tau(s)}^{(k-2)}}^2] \, ds
        \le \delta^2 h^2 + \beta^2 h^2 \, \mc E_{k-1}\,.
    \end{align*}
    Substituting this into~\cref{eq:ulmcpf1}, we obtain
    \begin{align*}
        \mc E_k \lesssim \delta^2 h^4 + \beta^2 h^4\, \mc E_{k-1}\,.
    \end{align*}
    Using $h \lesssim 1/\sqrt\beta$ and iterating this bound,
    \begin{align}\label{eq:ulmcpf3}
        \mc E_K
        &\lesssim \exp(-\Omega(K))\,\mc E_1 + \delta^2 h^4\,.
    \end{align}
    We must now bound $\mc E_1$. To do so, we note that
    \begin{align}\label{eq:ulmcpf2}
        \Ex[\norm{X_{nh+mh/M}^{(1)} - X_{nh}}^2]
        &= \Ex\Bigl[\Bigl\lVert \int_{nh}^{nh+mh/M} P_t^{(1)} \, dt \Bigr\rVert^2\Bigr]
        \le h\int_{nh}^{nh+mh/M} \Ex[\norm{P_t^{(1)}}^2] \, dt\,.
    \end{align}
    Also,
    \begin{align*}
        \Ex[\norm{P_t^{(1)}}^2]
        &\lesssim \Ex[\norm{P_{nh}}^2] + \Ex\Bigl[\Bigl\lVert \int_{nh}^t \{-s(X_{nh}) -\gamma P_s^{(1)}\} \, ds + \sqrt{2\gamma}\,(B_t - B_{nh}) \Bigr\rVert^2\Bigr] \\
        &\lesssim \Ex[\norm{P_{nh}}^2] + h^2 \,\Ex[\norm{s(X_{kh})}^2] +\gamma^2 h \int_{nh}^t \Ex[\norm{P_s^{(1)}}^2] \, ds \\
        &\qquad{} + \gamma \,\Ex[\norm{\tilde B_t - \tilde B_{nh}}^2] + \Ex\Bigl[\Bigl\lVert \int_{nh}^t (s(X_{\tau(s)}^{(K-1)}) - \nabla V(X_s)) \, ds \Bigr\rVert^2\Bigr] \\
        &\lesssim \mc P + h^2 \delta^2 + h^2 \mc G +\gamma^2 h \int_{nh}^t \Ex[\norm{P_s^{(1)}}^2] \, ds + \gamma dh + h^2 \Delta\,.
    \end{align*}
    In the above bound, we were careful to recall that we are working under $\mb Q$, for which $\tilde B$ is the Brownian motion (not $B$).
    Also, we have defined the following quantities:
    \begin{align*}
        \mc P \deq \sup_{t\in [0,T]} \Ex[\norm{P_t}^2]\,, \qquad \mc G \deq \sup_{t\in [0,T]} \Ex[\norm{\nabla V(X_t)}^2]\,,
    \end{align*}
    and
    \begin{align*}
        \Delta \deq \sup_{t\in [nh, (n+1)h]} \Ex[\norm{s(X_{\tau(t)}^{(K-1)}) -\nabla V(X_t)}^2]\,.
    \end{align*}
    Applying Gr\"onwall's inequality again,
    \begin{align*}
        \Ex[\norm{P_t^{(1)}}^2]
        &\lesssim \mc P + h^2 \delta^2 + h^2 \mc G + \gamma dh + h^2 \Delta\,.
    \end{align*}
    Substituting this into~\cref{eq:ulmcpf2},
    \begin{align*}
        \mc E_1
        &\lesssim h^2 \mc P + h^4 \delta^2 + h^4 \mc G + \gamma dh^3 + h^4 \Delta\,.
    \end{align*}
    Substituting this into~\cref{eq:ulmcpf3} now yields
    \begin{align*}
        \mc E_K
        &\lesssim \exp(-\Omega(K)) \,(h^2 \mc P + h^4 \mc G + \gamma dh^3 + h^4 \Delta) + \delta^2 h^4\,.
    \end{align*}

    Recalling the definition of $\Delta$ and from~\cref{eq:ulmcpf4} and~\cref{eq:ulmcpf5}, we have proven that
    \begin{align*}
        \Delta
        \lesssim \delta^2 + \beta^2 \,&\Bigl(\sup_{t\in [nh,(n+1)h]} \Ex[\norm{X_t - X_{\tau(t)}}^2] \\
        &\qquad{}+ \exp(-\Omega(K)) \,(h^2 \mc P + h^4 \mc G + \gamma dh^3 + h^4 \Delta) + \delta^2 h^4 \Bigr)\,.
    \end{align*}
    Using $h\lesssim 1/\sqrt\beta$, this yields
    \begin{align*}
        \Delta
        &\lesssim \delta^2 + \beta^2 \,\Bigl(\sup_{t\in [nh,(n+1)h]} \Ex[\norm{X_t - X_{\tau(t)}}^2] + \exp(-\Omega(K)) \,(h^2 \mc P + h^4 \mc G + \gamma dh^3) \Bigr)\,.
    \end{align*}
    We also note that
    \begin{align*}
        \Ex[\norm{X_t - X_{\tau(t)}}^2]
        &= \Ex\Bigl[\Bigl\lVert \int_{\tau(t)}^t P_s \, ds \Bigr\rVert^2\Bigr]
        \le \frac{h^2}{M^2} \, \mc P\,.
    \end{align*}
    The quantities $\mc P$, $\mc G$ are controlled via~\Cref{lem:ulmc_moment_bds}.
    Now assume that $\exp(-\Omega(K)) \le 1/M^4$, which only requires $K \gtrsim \log M$ for a sufficiently large absolute constant.
    When the dust settles,
    \begin{align*}
        \Delta
        &\lesssim \delta^2 + \frac{\beta^2 dh^2}{M^2} + \frac{\beta^2 h^2}{M^2}\,\bigl(1+\frac{\kappa}{M^2}\bigr)\DKL{\mu_0\river \pi}
    \end{align*}
    Substituting this into~\cref{eq:ulmcpf6}, and recalling that $\gamma \asymp \sqrt{\beta}$, we finally obtain
    \begin{align*}
        \DKL{\pi_T \river \mu_T}
        &\lesssim \frac{T}{\sqrt\beta}\,\Bigl(\delta^2 + \frac{\beta^2 dh^2}{M^2} + \frac{\beta^2 h^2}{M^2}\,\bigl(1+\frac{\kappa}{M^2}\bigr)\DKL{\mu_0\river \pi}\Bigr)\,.
    \end{align*}
    This completes the proof.
\end{proof}

We must complement the discretization bound with a continuous-time convergence result, which can be obtained from off-the-shelf results.
See~\citet[Lemma 5]{Zhangetal23ULMC} for a statement which is convenient for our setting (adapted from~\cite{Maetal21NesterovMCMC}, which in turn followed the original entropic hypocoercivity due to Villani~\citep{Vil09Hypo}; see also~\citet{Mon23HMC} for the corresponding result for idealized Hamiltonian Monte Carlo).

\begin{theorem}\label{thm:ent_hypo}
    Assume that $V$ is $\beta$-smooth and that $\pi^X\propto \exp(-V)$ satisfies the LSI with constant $\alpha$.
    Consider the functional
    \begin{align*}
        \mc F(\mu\;\|\; \pi)
        &\deq \DKL{\mu \river \pi} + \Ex_\mu\bigl[\bigl\lVert \mf M^{1/2}\, \nabla \log \frac{\mu}{\pi}\bigr\rVert^2\bigr]\,, \qquad \mf M \deq \begin{bmatrix}
            1/(4\beta) & 1/\sqrt{2\beta} \\
            1/\sqrt{2\beta} & 4
        \end{bmatrix} \otimes I\,.
    \end{align*}
    Then, for all $t\ge 0$,
    \begin{align*}
        \mc F(\pi_t \;\|\;\pi)
        &\le \exp\Bigl(-\frac{\alpha t}{10\sqrt{2\beta}}\Bigr)\,\mc F(\pi_0\;\|\;\pi)\,.
    \end{align*}
\end{theorem}

We are now ready to prove~\Cref{thm:ulmc_final}.

\medskip{}

\begin{proof}[Proof of~\Cref{thm:ulmc_final}]
    Let us show that $\mu_0 = \mc N(x^\star, \beta^{-1} I) \otimes \mc N(0, I)$ satisfies
    \begin{align*}
        \mc F(\mu_0\;\|\;\pi)
        &\le \frac{d}{2}\,(2+\log \kappa)\,.
    \end{align*}
    From~\Cref{cor:lmc_guarantee}, we know that $\DKL{\mu_0\river\pi} \le \frac{d}{2}\log\kappa$.
    Also,
    \begin{align*}
         \Ex_{\mu_0}\bigl[\bigl\lVert \mf M^{1/2}\, \nabla \log \frac{\mu_0}{\pi}\bigr\rVert^2\bigr]       
         &= \frac{1}{4\beta}\, \Ex_{\mc N(x^\star, \beta^{-1} I)}\bigl[\bigl\lVert \nabla \log \frac{\mc N(x^\star, \beta^{-1} I)}{\pi^X}\bigr\rVert^2\bigr] \\
         &= \frac{1}{4\beta}\, \Ex_{x\sim\mc N(x^\star, \beta^{-1} I)}\bigl[\bigl\lVert \nabla V(x) - \frac{\beta}{2}\,(x-x^\star)\bigr\rVert^2\bigr] \\
         &\le \frac{1}{2\beta}\, \Ex_{x\sim\mc N(x^\star, \beta^{-1} I)}\bigl[\norm{\nabla V(x) - \nabla V(x^\star)}^2 +\frac{\beta^2}{4}\,\norm{x-x^\star}^2\bigr] \\
         &\le \beta\, \Ex_{x\sim\mc N(x^\star, \beta^{-1} I)}[\norm{x -x^\star}^2] \le d\,.
    \end{align*}
    The initialization bound follows.

    The setting of parameters is such that from~\Cref{thm:ulmc_discretization} and~\Cref{thm:ent_hypo} respectively, we have $\DKL{\pi_{Nh} \river \mu_{Nh}} \lesssim \varepsilon^2$ and $\DKL{\pi_{Nh} \river \pi} \lesssim \varepsilon^2$.
    The result now follows from Pinsker's inequality and the triangle inequality for TV\@.
\end{proof}

\section{Proofs for sampling from discrete distributions}\label{app:discrete}

We begin with the proof of~\Cref{lem:computation}.

\medskip{}

\begin{proof}[Proof of~\Cref{lem:computation}]
The first two statements are from \cite{AHLVXY22}.
We only need to verify the last statement. We only need to show that we can approximate $\mean{\tau_{z} \mu}$ for all $z \in \R^n$, given the oracle for the Laplace transform of $\mu$. Since $\mu $ is supported on the hypercube, we can rewrite the $j$-th entry of $\mean {\tau_{z} \mu}$ in term of Laplace transforms of $\mu$, i.e.,
\begin{align*}
    (\mean {\tau_{z} \mu})_j
    &= 2\,\tau_z\mu(x_j = +) - 1
    = \frac{2\sum_{x\in \{\pm\}^n, \; x_j=+} \exp(\langle z,x\rangle) \, \mu(x)}{\sum_{x\in \{\pm\}^n} \exp(\langle z, x\rangle)\,\mu(x)} -1 \\
    &= \frac{2\exp(z_j) \sum_{x\in \{\pm\}^n, \; x_j=+} \exp(\langle z_{-j} ,x_{-j}\rangle) \, \mu(x)}{\sum_{x\in \{\pm\}^n} \exp(\langle z, x\rangle)\,\mu(x)} -1 \\
    &= 2\exp\bigl(z_j + \L_\mu{z^+} - \L_\mu{z}\bigr)-1\,,
\end{align*}
where $z^+ $ (resp.\ $z^-$) is a vector with all entries equal to $z$ except for the $j$-th entry being $+\infty$ (resp.\ $-\infty$).
Using the oracle, we can compute $ \hat{A}_+$ s.t. $ \abs{\hat{A}_+ - (\L_{\mu}{ z^+} - \L_{\mu}{z} )}\leq O(\varepsilon).$
Thus,
\begin{align*}
    &\abs{2 \exp(z_j + \hat{A}_+) - 1 - (\mean {\tau_z \mu})_j} \\
    &\qquad = 2 \exp(z_j)\exp(\L_\mu{z^+} - \L_\mu{z})\, \bigl\lvert \exp\bigl(\hat A_+ - (\L_\mu{z^+} - \L_\mu{z})\bigr)-1\bigr\rvert \\
    &\qquad \leq O(\varepsilon) \exp(z_j)\exp(\L_\mu{z^+} - \L_\mu{z})
    = O(\varepsilon) \,\frac{(\mean{\tau_{z}\mu})_j + 1}{2}
    = O(\varepsilon)
\end{align*} 
where the inequality follows from $ \exp(x)-1 \leq 2x$ for $x\in [0,1/2)$.
We use $n$ machines, each of which computes one entry of $\mean{\tau_z \mu}$ using $2$ oracle calls and $O(1)$ parallel iterations.
The estimated score function $s$ satisfies $\norm{s(y) - \nabla V(y)} \lesssim \sqrt{\frac{n}{c} \,\varepsilon^2} = \delta$.
\end{proof}

We also need another initialization lemma, since~\Cref{cor:lmc_guarantee} requires knowledge of the minimizer of $V$ which is not necessarily the case for the present application.

\begin{lemma}\label{lem:initialization bound general}
    Let $\mu_0 = \Normal{y, \sigma^2 I }$ for some fixed $y \in \R^n$ and $\sigma^2 > 0$. If $\pi \propto \exp(-V)$ with $\nabla^2 V \preceq \beta I$, then 
    \[\DKL {\mu_0 \river \pi} \leq V(y) +\log Z+ \frac{n}{2}\, (\beta \sigma^2 - \log (2\pi e\sigma^2) )\]
    where $Z = \int  \exp(-V(x))\, dx$. 
\end{lemma}
\begin{proof}
    By smoothness,
    $V(x) \leq V(y) + \langle\nabla V(y), x-y\rangle + \frac{\beta}{2}\, \norm{x-y}^2$,
    thus
    \[\E_{x\sim \mu_0} {V(x) }\leq V(y) + \langle \nabla V(y), \E_{x\sim \mu_0} {x- y} \rangle + \frac{\beta}{2}  \E_{x\sim \mu_0} {\norm{x-y}^2} = V(y) + \frac{\beta \sigma^2 n}{2}\]
    and 
    \begin{align*}
        \DKL { \mu_0\river \pi}
        &= \E_{x\sim \mu_0} {\log \mu_0(x) + V(x) + \log Z } = -\frac{n}{2} \log (2\pi e\sigma^2) + V(y) + \frac{\beta\sigma^2 n}{2} + \log Z\,,
    \end{align*}
    which is the desired bound.
\end{proof}

\begin{lemma}\label{lem:initialization bound for discrete convolve with Gaussian}
    Consider a density function $ \nu: \set{\pm 1}^n \to \R_{\geq 0}.$ Let $\pi = \nu\ast \Normal{0,c I}$ and $\mu_0 =\Normal{0,c I}$.
    Then,
    \[ \DKL {\mu_0 \river \pi} \leq \frac{n}{2c}\,.  \]
\end{lemma}
\begin{proof}
We can write
\[\pi(y) = (2\pi c)^{-n/2} \sum_{x\in \set{\pm 1}^n} \nu(x) \exp\bigl(-\frac{\norm{y-x}^2}{2c}\bigr)\,. \]
This distribution is normalized so that $Z = 1$, and
\[\pi(0)  = (2\pi c)^{-n/2} \sum_{x\in \set{\pm 1}^n} \nu(x) \exp\bigl(-\frac{n}{2c}\bigr) = (2\pi c)^{-n/2} \exp\bigl(-\frac{n}{2c}\bigr)\,. \]
Thus, $V(0) = -\log \pi(0) =\frac{n}{2} \log (2\pi c) + \frac{n}{2c}  . $ 
 By~\Cref{lem:computation}, $\nabla^2 V \preceq I/c$. Thus, we can apply \cref{lem:initialization bound general} with $\beta = c^{-1}$ and $\sigma^2 = c.$ Rearranging gives the desired inequality.
    \end{proof}

\begin{proof}[Proof of~\Cref{thm:rnc discrete sampling}]
Let $c $ be such that $ \cov{\tau_y \mu} \preceq \frac{c}{2} I$  for all $y \in \R^n$.
Suppose we have two executions of \cref{alg:sampling to counting reduction}: one using the approximate continuous sampling algorithm resulting in $w_0,\dots,w_T$, and one using exact samples resulting in $w_0',\dots,w_T'$. Note that $w_{i} = w_{i-1} + x_i/c$ where $x_i$ is the output of~\Cref{alg:main} on input $ \pi = \tau_{w_{i-1}}\mu \ast \Normal{0,cI} $ and $ w'_i = w'_{i-1} + x'_i/c$ where $x'_i \sim \tau_{w'_{i-1}}\mu \ast \Normal{0,cI}. $ We choose the parameter of~\Cref{alg:main} so that
\begin{align*}
    \dTV{\law(x_i), \tau_{w_{i-1}}\mu \ast \Normal{0,cI}} \leq \eta
\end{align*}
for some $\eta$ to be specified later.

%

Recall that the total variation distance is also characterized as the smallest probability of error when we couple two random variables according to the two measures, i.e.,
\begin{align*}
    \dTV{\rho_1,\rho_2} = \inf\bigl\{\Pi(X_1 \ne X_2) \bigm\vert \Pi~\text{is a coupling of}~(\rho_1,\rho_2)\bigr\}\,.
\end{align*}
On the first iteration, we can couple $x_1$ with $x_1'$ so that they are equal to each other with probability at least $1-\eta$. If $x_1 = x_1'$, then $w_1 = w_1'$, and repeating the argument on this event we can couple $x_2$ to $x_2'$ so that $x_2 = x_2'$ with probability at least $1-\eta$.
After $T$ iterations, by the union bound, we have $w_T = w_T'$ with probability at least $1-T\eta$.

By triangle inequality, the data-processing inequality, and~\Cref{lem:tv-close},
\begin{align*}
   \dTV{\law(\sign{w_T}), \mu}
   &\leq \dTV{\law(\sign{w_T}), \law(\sign{w'_T})} + \dTV{\law(\sign{w'_T}),\mu} \\
   &\leq T\eta + \varepsilon/2\,,
\end{align*}
provided we choose $T = \Theta(c \log (n/\varepsilon))$ so that $\dTV{\law(\sign{w'_T}),\mu} \leq \varepsilon/2$.
We then choose $\eta = \varepsilon/(2T)$, which ensures that $\dTV{\law(\sign{w_T}), \mu} \le \varepsilon$.

In each iteration of the ``for'' loop in~\Cref{alg:sampling to counting reduction}, we want to approximately sample from $\pi = \tau_{w'_{i-1}} \mu \ast \Normal{0,cI}, $ which is $(2c)^{-1}$-strongly log concave and $c^{-1}$-log-smooth 
by~\Cref{lem:computation}. By~\Cref{lem:initialization bound for discrete convolve with Gaussian}, $\DKL{\mu_0\river \pi} \leq \poly(n)$ for $\mu_0 = \Normal{0,c I}$. 
Thus, by~\Cref{thm:sample-main}, to sample $x'_i$ such that $\dTV{\law(x'_i),\tau_{w'_{i-1}} \mu \ast \Normal{0,cI}} \leq O(\varepsilon/(c\log(n/\varepsilon)))$,~\Cref{alg:main}  uses $P = O(\log^2(cn/\varepsilon))$ parallel iterations, $M = \widetilde O(c^2 n/\varepsilon^2)$ processors, and $M P = \widetilde{O}(c^2 n/\varepsilon^2)$ $\delta$-approximate gradient evaluations with $\delta = \Theta(\varepsilon/\sqrt{c})$. By~\Cref{lem:computation}, each gradient evaluation can be implemented using $O(n)$ processors, $O(1)$ parallel iterations, and $O(n)$ total calls to $O(\delta \sqrt c/n) = O(\varepsilon /n)$-approximate Laplace transform oracles.  

Hence,~\Cref{alg:sampling to counting reduction} takes $P T = O(c  \log^3(cn/\epsilon)) $ parallel iterations, $M = \widetilde O(c^2 n^2/\varepsilon^2)$ processors, and $\widetilde{O}(c^2 n^2/\varepsilon^2)$ total calls to $O(\varepsilon/n)$-approximate Laplace transform oracles.
\end{proof}

	\PrintBibliography
\end{document}